\documentclass[a4paper,USenglish]{lipics-v2018}

\usepackage{microtype}

\newcommand{\bone}[1]{\ensuremath{{\mathcal{K}}_b(#1)}}
\newcommand{\sone}[1]{\ensuremath{{\mathcal{K}}_s(#1)}}

\newcommand{\two}[1]{\ensuremath{{\mathcal{K}}_2(#1)}}
\newcommand{\tone}{type-\textsc{i}}
\newcommand{\ttwo}{type-\textsc{ii}}

\usepackage{hyperref}
\usepackage{comment}
\usepackage{booktabs}
\usepackage{todonotes}

\usepackage{amssymb,amsmath,amsthm,amstext}

\usepackage{graphicx,enumerate}

\usepackage{tikz} 
\usetikzlibrary{graphs,graphs.standard,quotes}
\usetikzlibrary{shapes.geometric}
\tikzstyle{myv} = [{circle,blue,draw,fill=black!50,inner sep=1pt}]
\usetikzlibrary{arrows,shapes,positioning}
\usetikzlibrary{calc}

\newcommand{\NP}{\textsf{NP}}
\newcommand{\FPT}{\textsf{FPT}}

\nolinenumbers

\newtheorem{redrule}{Rule}
\newtheorem{proposition}[theorem]{Proposition}

\title{A Polynomial Kernel for Diamond-Free Editing}

\titlerunning{A Polynomial Kernel for Diamond-Free Editing}

\author{Yixin Cao}{Department of Computing, Hong Kong Polytechnic University, China.}{yixin.cao@polyu.edu.hk}{0000-0002-6927-438X}{}
\author{Ashutosh Rai}{Department of Computing, Hong Kong Polytechnic University, China.}{ashutosh.rai@polyu.edu.hk}{}{}
\author{R. B. Sandeep}{Institute for Computer Science and Control, Hungarian Academy of Sciences, (MTA SZTAKI,) Hungary; and Indian Institute of Technology Dharwad, India.}{sandeeprb@iitdh.ac.in}{}{}
\author{Junjie Ye}{Department of Computing, Hong Kong Polytechnic University, China.}{junjie.ye@polyu.edu.hk}{0000-0003-3924-008X}{}

\authorrunning{Cao, Rai, Sandeep, Ye}

\Copyright{Yixin Cao, Ashutosh Rai, R. B. Sandeep, and Junjie Ye}



\category{}

\relatedversion{}

\supplement{}

\funding{Supported in part by the Hong Kong Research Grants Council (RGC) under grant PolyU 252026/15E, the National Natural Science Foundation of China (NSFC) under grant 61572414, and the European Research Council (ERC) under grant 725978.}

\acknowledgements{}

\EventEditors{John Q. Open and Joan R. Access}
\EventNoEds{2}
\EventLongTitle{42nd Conference on Very Important Topics (CVIT 2016)}
\EventShortTitle{CVIT 2016}
\EventAcronym{CVIT}
\EventYear{2016}
\EventDate{December 24--27, 2016}
\EventLocation{Little Whinging, United Kingdom}
\EventLogo{}
\SeriesVolume{42}
\ArticleNo{23}
\hideLIPIcs  

\begin{document}

\maketitle

\begin{abstract}
  Given a fixed graph $H$, the $H$-free editing problem asks whether we can edit at most $k$ edges to make a graph contain no induced copy of $H$. We obtain a polynomial kernel for this problem when $H$ is a diamond. The incompressibility dichotomy for $H$ being a 3-connected graph \cite{cai2015incompressibility} and the classical complexity dichotomy \cite{AravindSS17} suggest that except for $H$ being a complete/empty graph, $H$-free editing problems admit polynomial kernels only for a few small graphs $H$. Therefore, we believe that our result is an essential step toward a complete dichotomy on the compressibility of $H$-free editing. Additionally, we give a cubic-vertex kernel for the diamond-free edge deletion problem, which is far simpler than the previous kernel of the same size for the problem.
\end{abstract}

\section{Introduction}
\label{sec:intro}

A graph modification problem asks whether one can apply at most $k$ modifications to a graph to make it satisfy certain properties. By modifications we usually mean additions and/or deletions, and they can be applied to vertices or edges. 
Although other modifications are also considered, most results in literature are on vertex deletion and the following three edge modifications: edge deletion, edge completion, and edge editing (deletion/completion). 

As usual, we use $n$ to denote the number of vertices of the input graph.  For each graph modification problem, one may ask three questions: (1) Is it \NP-complete? (2) Can it be solved in time $f(k) \cdot n^{O(1)}$ for some function $f$, and if yes, what is the (asymptotically) best $f$? (3) Does it have a polynomial kernel?  The first question concerns the classic complexity, while the other two are about the parameterized complexity \cite{flum-grohe-06, downey-13}.   With parameter $k$, a problem is \emph{fixed-parameter tractable (\FPT)} if it can be solved in time $f(k)\cdot n^{O(1)}$ for some function $f$.   On the other hand, given an instance $(G, k)$, a {\em kernelization algorithm} produces in polynomial time an equivalent instance $(G', k')$---$(G, k)$ is a yes-instance if and only if $(G', k')$ is a yes-instance---such that $k' \leq k$.  It is a \textit{polynomial kernel} if the size of $G'$ is bounded from above by a polynomial function of $k$.

 For hereditary properties, a classic result of Lewis and Yannakakis \cite{lewis1980node} states that all the vertex deletion problems are either \NP-hard or trivial.
In contrast, the picture for edge modification problems is far murkier.  Earlier efforts for edge deletion problems \cite{Yannakakis81edge, el1988complexity}, though having produced fruitful concrete results, shed little light on a systematic answer, and it was noted that such a generalization is difficult to obtain. 

A basic and ostensibly simple case of graph modification problems is to make the graph $H$-free, where $H$ is a fixed graph on at least two vertices. (We say that a graph is $H$-free if it does not contain $H$ as an induced subgraph.)  For this special case, all the three questions have been satisfactorily answered for vertex deletion problems, at least in the asymptotic sense.  All of them are \NP-complete and \FPT,---indeed, $H$-free vertex deletion problems admit simple $|V(H)|^k \cdot n^{O(1)}$-time algorithms~\cite{cai96fixed}.  On the other hand, the reduction of Lewis and Yannakakis~\cite{lewis1980node} excludes subexponential-time algorithms ($2^{o(k)} \cdot n^{O(1)}$-time algorithms) assuming the exponential time hypothesis (\textsc{eth}) \cite{impagliazzo-01-eth}.  Further, as observed by Flum and Grohe \cite{flum-grohe-06}, the sunflower lemma of Erd\H{o}s and Rado \cite{erdos-60-sunflower} can be used to produce polynomial kernels for $H$-free vertex deletion problems.

Even restricted to this very simple case, edge modification problems remain elusive.  Significant efforts have been devoted to an ongoing program that tries to answer these questions in a systematic way, and promising progress has been reported in literature.
Recently, Aravind et al.~\cite{AravindSS17} gave a complete answer to the first question: The $H$-free editing problem is \NP-complete if and only if $H$ contains at least three vertices. They also excluded subexponential-time algorithms for the \NP-complete $H$-free edge modification problems, assuming \textsc{eth}.  Noting that $H$-free edge modification problems can always be solved in  $2^{O(k)} \cdot n^{O(1)}$ time~\cite{cai96fixed}, we are left with the third problem, the existence of polynomial kernels.

Some of the $H$-free graph classes are important for their own structural reasons, e.g., most notably, cluster graphs and cographs, which are $P_3$-free graphs and $P_4$-free graphs respectively; hence the edge modification problems toward them have been well-studied  \cite{cao-12-kernel-cluster-editing, guillemot2013non}.  (Note that edge modification problems to $P_2$-free graphs, i.e., independent sets, are trivial.)  Given the simplicity of $H$-free edge modification problems, and the naive \FPT{} algorithms for them, it may sound shocking that many of them do \textit{not} admit polynomial kernels.  
Indeed, the earliest incompressibility results of graph modification problems, by Kratsch and Wahlstr\"om \cite{kratsch2013two}, are on $H$-free edge modification problems.  
Guillemot et al.~\cite{guillemot2013non} excluded polynomial kernels for $H$-free edge deletion problems when $H$ is a path of length at least seven or a cycle of length at least four.   An influential result of Cai and Cai~\cite{cai2015incompressibility} furnishes a dichotomy on the compressibility of $H$-free edge modification problems when $H$ is a path, a cycle, or a $3$-connected graph.  


We tend to believe that $H$-free edge modification problems admitting polynomial kernel are the exceptions.  Our exploration suggests that graphs on four vertices play the pivotal roles if we want to fully map the territory.  Let $\overline H$ be the complement graph of $H$.  Then the $H$-free edge deletion problem is equivalent to the $\overline H$-free edge completion problem, while the edge editing problems are the same for $H$-free and $\overline H$-free graphs.  We are thus focused on the four-vertex graphs (Figure~\ref{fig:diamond}); see Table~\ref{table:1} for a summary of compressibility results of $H$-free edge modification problems when $H$ is one of them.  We conjecture that $H$-free edge modification problems, when $H$ being claw or paw, admit polynomial kernels. 




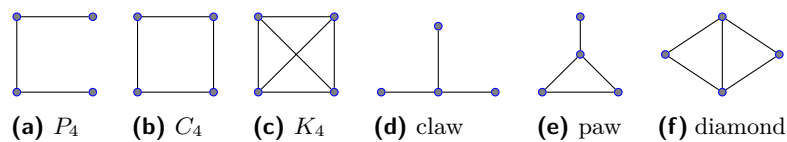
\begin{figure}[h]
  \centering
\begin{subfigure}[b]{.08\linewidth}
  \centering
    \begin{tikzpicture}[every node/.style={myv},scale=.5]
      \node (a) at (-1,0) {};
      \node (c) at (1,0) {};
      \node (b) at (-1,2) {};
      \node (d) at (1,2) {};
      \draw (d) -- (b) -- (a) -- (c);
    \end{tikzpicture}
\caption{$P_4$}\label{fig:p4}
\end{subfigure}  
\quad
\begin{subfigure}[b]{.08\linewidth}
  \centering
    \begin{tikzpicture}[every node/.style={myv}, scale=.5]
      \node (a) at (-1,0) {};
      \node (c) at (1,0) {};
      \node (b) at (-1,2) {};
      \node (d) at (1,2) {};
      \draw (a) -- (b) -- (d) -- (c) -- (a);
    \end{tikzpicture}
\caption{$C_4$}\label{fig:c4}
\end{subfigure}
\quad
\begin{subfigure}[b]{.08\linewidth}
  \centering
    \begin{tikzpicture}[every node/.style={myv}, scale=.5]
      \node (a) at (-1,0) {};
      \node (d) at (1,0) {};
      \node (b) at (-1,2) {};
      \node (c) at (1,2) {};
      \draw (a) -- (b) -- (c) -- (d) -- (a) -- (c) (b) -- (d);
    \end{tikzpicture}
\caption{$K_4$}\label{fig:k4}
\end{subfigure}
\quad
\begin{subfigure}[b]{.12\linewidth}
  \centering
    \begin{tikzpicture}[every node/.style={myv},scale=.25]
      \node (a1) at (-3., 0) {};
      \node (v) at (0, 0) {};
      \node (b1) at (3., 0) {};
      \node (c) at (0,3.5) {};
      \draw[] (a1) -- (v) -- (b1);
      \draw[] (v) -- (c);
    \end{tikzpicture}
\caption{claw}\label{fig:claw}
\end{subfigure}
\quad
\begin{subfigure}[b]{.08\linewidth}
  \centering
    \begin{tikzpicture}[every node/.style={myv},scale=.25]
      \node (s) at (0,4) {};
      \node (a1) at (-2,0) {};
      \node (b1) at (2,0) {};
      \node (c) at (0,2) {};
      \draw[] (c) -- (s);
      \draw[] (a1) -- (c) -- (b1) -- (a1);
    \end{tikzpicture}
\caption{paw}\label{fig:paw}
\end{subfigure}
\quad
\begin{subfigure}[b]{.12\linewidth}
  \centering
    \begin{tikzpicture}[every node/.style ={myv}, scale = .5]
      \node (x) at (0, 1) {};
      \node (u) at (-1.5, 0) {};
      \node (v) at (1.5, 0) {};
      \node (y) at (0, -1) {};
      \draw (u) -- (x) -- (v) -- (y) -- (u);
      \draw (x) -- (y);
    \end{tikzpicture}
\caption{diamond}\label{fig:diamond}
\end{subfigure}
  \caption{Graphs on four vertices (their complements are omitted).}
  \label{fig:diamond}
\end{figure}

\begin{table}[ht]
  \centering
  \begin{tabular}{l l l l}
    \toprule
    $H$ & deletion & completion & editing
    \\ \midrule
    $K_4$ & $O(k^4)$ \cite{cai2012polynomial} & trivial & $O(k^4)$ \cite{cai2012polynomial}
    \\
    $P_4$ & $O(k^3)$ \cite{guillemot2013non} & $O(k^3)$ \cite{guillemot2013non} & $O(k^3)$ \cite{guillemot2013non}
    \\
    diamond & $O(k^3)$ \cite{sandeep2015parameterized} & trivial & $O(k^8)$ [this paper]
    \\\midrule
    claw & unkown & unkown & unkown
    \\
    paw & unkown & unkown & unkown
    \\\midrule
    $C_4$ & no \cite{guillemot2013non} & no \cite{guillemot2013non} & no \cite{guillemot2013non}
    \\
    \bottomrule
  \end{tabular}
  \caption{The compressibility results of $H$-free edge modification problems for $H$ being four-vertex graphs.  Note that every result holds for the complement $H$; e.g., the answers are also no when $H$ is $2 K_2$.}
  \label{table:1}
\end{table}


We show a polynomial kernel for the diamond-free editing problem.  Our observations also lead to a cubic-vertex kernel for the diamond-free edge deletion problem, which is far simpler than the previous kernel of the same size~\cite{sandeep2015parameterized}.  Formally, these two problems are defined as following. 
\begin{quote}
  {Diamond-free editing problem:} Given an input graph $G$, can we edit (add/delete) at most $k$ edges to make it diamond-free?
\end{quote}

\begin{quote}
  {Diamond-free edge deletion problem:} Given an input graph $G$, can we delete at most $k$ edges to make it diamond-free?
\end{quote}
On the other hand, the completion problem to diamond-free graphs is trivial: For each diamond, there is no choice but adding the only missing edge.

Our key observations are on the maximal cliques.  A graph $G$ is diamond-free if and only if every two maximal cliques of $G$ share at most one vertex.  We say that a maximal clique is of \emph{type \textsc{i}} if it shares an edge with another maximal clique, or \emph{type \textsc{ii}} otherwise.  It is not hard to see that to make a graph diamond-free, we should never delete edges from a sufficiently large clique.  We thus put the maximal cliques of $G$ into three categories, small type \textsc{i}, big type \textsc{i}, and type \textsc{ii}.
It turns out that a vertex participates in a diamond if and only if it is in a maximal clique of type \textsc{i}, and the small type-\textsc{i} maximal cliques are crucial for the problem.


The first phase of our algorithm comprises two routine reduction rules for edge editing problems.  If a (non-)edge participates in $k + 1$ or more diamonds that pairwise share only this (non-)edge, then it has to be in a solution of size at most $k$.
(This is exactly the reason why no edge is deleted from a ``large'' clique.)
If there exists such an edge/non-edge, we delete/add it.  We may henceforth assume that these two simple rules have been exhaustively applied.
We are able to show that the ends of an edge added by a minimum solution must be from some small maximal cliques of type \textsc{i}\@. The situation for deleted edges is slightly more complex. 
The two ends of a deleted edge are either in a small maximal clique of type \textsc{i}, or in a maximal clique of type \textsc{ii}\@.  In the second case, the maximal clique has to intersect some small maximal clique of type \textsc{i}\@.

The second phase of our algorithm uses three nontrivial reduction rules to delete irrelevant vertices.  To analyze the size of the kernel, we bound the number of vertices that are (a) in small type-\textsc{i} maximal cliques only, (b) in big type-\textsc{i} maximal cliques but not in any small type-\textsc{i} maximal clique, and (c) only in type-\textsc{ii} maximal cliques.
First, we show an upper bound on the number of type-\textsc{i} maximal cliques. This immediately bounds the number of vertices in part (a), because each small type-\textsc{i} maximal clique has a bounded size. For part (b), the focus now is to bound the sizes of big maximal cliques of type \textsc{i}.  We introduce another reduction rule to delete certain ``{private vertices}'' from them.  On the other hand, the pattern of vertices shared by big maximal cliques is very limited. We are thus able to bound the number of vertices in part (b), and we are left with part (c). We correlate a maximal clique $K$ of type \textsc{ii} with small maximal cliques of type \textsc{i}: we would touch $K$ only because it had become type \textsc{i} after some operation, and this operation has to be an edge addition. Recall that an edge can only be added between two vertices in part (a). For each pair of them, we can build a blocker of $O(k^2)$ vertices from part (c). One more reduction rule is introduced to remove all vertices behind the blockers. Together with the bound of vertices in part (a), this bounds the number of vertices in part (c).  
%

Putting everything together, we obtain the main result of this paper.

\begin{theorem}\label{thm:main}
 The diamond-free editing problem has a kernel of $O(k^8)$ vertices. 
\end{theorem}

In passing we would like to mention that the reduction rules in the second phase of our algorithm only delete vertices, and none of the deleted vertices is from a small maximal clique of type \textsc{i}.  Hence, the main structure of the graph will not change, and this allows us to run them only once, and more importantly, we do not need to re-run the reduction rules in the first phase. 

Before we conclude this section, let us have some remarks on future work.  The main purpose of this paper is to take a step toward a complete dichotomy on the compressibility of $H$-free edge modification problems.  Although our kernel for the diamond-free editing problem does not seem to be tight, we do not consider obtaining a smaller kernel a pressing issue.  Instead, we would ask for the existence of polynomial kernels for the $H$-free editing problem in general, and for $H$ being claw or paw in particular.  

Another question one may ask for graph modification problems is: (4) Does it admit a constant-ratio approximation algorithm?  Again, the answer is simply yes for $H$-free vertex deletion problems. 
Once an induced copy of $H$ is found, we delete all its vertices, which achieves ratio $|V(H)|$.  However, this simple algorithm breaks down for edge modification problems: An edge added or deleted to fix an erstwhile $H$ may introduce new one(s).  This kind of propagations are the core difficulty of these problems \cite{cai2015incompressibility, bliznets-16-hardness-h-free}.  Hence, both the approximability and the compressibility seem to hinge on whether the propagations can be tamed, though maybe in different senses.

\medskip
\noindent {\em Outline of the Paper.} After two simple reduction rules, Section~\ref{sec:max-cliques} studies maximal cliques in the graph on which these two rules are not applicable, and how a minimum solution may impact them.
Section~\ref{sec:kernel} presents the main reduction rules, and establishes the bounds on the numbers of vertices in different categories, thereby proving the main theorem of the paper. 
Using the properties established in previous sections, the last section presents a very simple kernelization algorithm for the diamond-free edge deletion problem.

\section{Maximal cliques}
\label{sec:max-cliques}

All graphs discussed in this paper are undirected and simple.  A graph $G$ is given by its vertex set $V(G)$ and edge set $E(G)$.  
The \emph{neighborhood} of a vertex $v$ in a graph $G$, denoted by $N_G(v)$, consists of all the vertices adjacent to $v$ in $G$. We extend this to a set $S \subseteq V(G)$ of vertices by defining the neighborhood $N_G(S)$ of $S$ as $(\bigcup_{v \in S}  N_G(v)) \setminus S$.
For a set $U\subseteq V(G)$ of vertices, we denote by $G[U]$ the subgraph induced by $U$, whose vertex set is $U$ and whose edge set comprises all edges of $G$ with both ends in $U$.  We use $G - v$, where $v$ is a vertex of $G$, as a shorthand for $G[V(G)\setminus \{v\}]$.  
In a diamond, we refer to the edge between the two degree-three vertices as the \textit{cross edge}, and the only non-edge the \textit{missing edge}.

For a set $E_+$ of edges, we denoted by $G + E_+$ the graph obtained by adding edges in $E_+$ to $G$,---its vertex set is still $V(G)$ and its edge set becomes $E(G)\cup E_+$.  The graph $G - E_-$ is defined analogously.
Throughout the paper we always tacitly assume $E_+\cap E(G)=\emptyset$ and $E_-\subseteq E(G)$; hence $E_+$ and $E_-$ are disjoint.
A \emph{solution} of an instance $(G, k)$ consists of a set $E_+$ of added edges and a set $E_-$ of deleted edges such that $G + E_+ - E_-$ is diamond-free and $|E_+\cup E_-| \le k$.
We use $E_\pm$ as a shorthand for $E_+\cup E_-$, and there should be no ambiguities: $E_+ = E_\pm\setminus E(G)$ and $E_- = E_\pm\cap E(G)$. We also use $G \triangle E_\pm$ as a shorthand for $G + E_+ - E_-$. 

We start from a routine step for edge editing problems.  If an edge $u v$ participates in $k + 1$ or more diamonds that pairwise share only this edge, then it has to be in any solution of size at most $k$. The following two reduction rules, taking care of the cases $u v$ being the missing edge and being the cross edge respectively, would suffice for our purpose. The correctness of them is straightforward: If we do not add/delete $uv$, then we have to delete/add at least $k + 1$ edges. 

\begin{redrule}
  \label{rul:nonedge-sunflower}
  If there exist a non-edge $uv$ and $2k + 2$ distinct vertices $x_1, y_1, \ldots, x_{k + 1}, y_{k + 1}$ in $N(u) \cap N(v)$ such that $x_i y_i\in E(G)$ for all $1\le i \le k + 1$, then add $uv$ and decrease $k$ by one. 
\end{redrule}

\begin{redrule}
  \label{rul:edge-sunflower}
  If there exist an edge $uv$ and $2k + 2$ distinct vertices $x_1, y_1, \ldots, x_{k + 1}, y_{k + 1}$ in $N(u) \cap N(v)$ such that $x_i y_i\not\in E(G)$ for all $1\le i \le k + 1$, then delete $uv$ and decrease $k$ by one.
\end{redrule}

For a pair of nonadjacent vertices $u, v$, whether Rule~\ref{rul:nonedge-sunflower} is applicable to $u v$ can be decided by finding a maximum matching in $G[N(u)\cap N(v)]$.  Likewise, for $u v\in E(G)$, we can find a maximum matching in the complement graph of $G[N(u)\cap N(v)]$.  Therefore, Rules~\ref{rul:nonedge-sunflower} and~\ref{rul:edge-sunflower} can be applied in polynomial time.
We call an instance $(G,k)$ \emph{reduced} if neither of Rules~\ref{rul:nonedge-sunflower} and~\ref{rul:edge-sunflower} is applicable to it.  In the rest, we will focus on reduced instances.
A similar idea as the two rules enables us to exclude some (non-)edges from consideration.
\begin{proposition}\label{pro:untouchable-edge}
  Let $E_\pm$ be a solution to a yes-instance $(G, k)$.  A (non-)edge $u v$ cannot be in $E_\pm$ if
  \begin{enumerate}[(i)]
  \item $uv \in E(G)$ and there are $k+1$ pairwise adjacent vertices in $N(u)\cap N(v)$; or
  \item $uv \not\in E(G)$ and there are $k+1$ pairwise nonadjacent vertices in $N(u)\cap N(v)$.
  \end{enumerate}
\end{proposition}
\begin{proof}
  Suppose for contradiction to assertion (i) that $uv\in E_-$.  Let $K$ be a set of $k+1$ pairwise adjacent vertices in $N(u)\cap N(v)$, and let $X \subseteq K$ be the set of vertices $x$ with $xu$ or $xv$ in $E_-$. Vertices in $K \setminus X$ remain adjacent to both $u$ and $v$ in $G - E_-$.  Since the final graph is diamond-free, $E_-$ must contain all edges among $K \setminus X$.  Therefore, 
  \[
    |E_-| \ge 1 + |X| + (|K \setminus X| - 1) = k + 1,
  \]
  which is impossible because $E_\pm$ is a solution to $(G, k)$.
  The argument for assertion (ii) is similar and hence omitted.
\end{proof}


\begin{proposition}\label{pro:edited-neighborhood}
  Let $(G, k$) be a reduced yes-instance.  For any (non-)edge $uv$ in a solution of $(G, k$), the cardinality of $N(u) \cap N(v)$ is at most $3k$.
\end{proposition}
\begin{proof}
  We consider only $u v\in E_-$, and the argument for $uv\in E_+$ is similar and omitted.
  Let $W = N(u) \cap N(v)$;  we find a maximum matching in the complement graph of $G[W]$, and let $W'$ be the ends of the edges in the matching.  Since Rule~\ref{rul:edge-sunflower} is not applicable (to $u v$), $|W'| \le 2 k$.  There cannot be non-edges between vertices in $W\setminus W'$; then by Proposition~\ref{pro:untouchable-edge}(i), the size of $W\setminus W'$ is at most $k$.  Therefore, $|W| \le 3k$.  
\end{proof}

Our algorithm will be mostly concerned with maximal cliques.  According to Proposition~\ref{pro:untouchable-edge}(i), a maximal clique on $k + 3$ or more vertices cannot be touched by a minimum solution ``{from inside},'' but it  may be touched ``{from outside}''---i.e., edges may be added between it and other vertices.   
We call a maximal clique \textit{big} if it contains at least $3k + 2$ vertices, and \textit{small} otherwise.  The bigness will prevent a maximal clique from being touched from outside.
\begin{lemma}\label{lem:bigT1}
  Let  $(G,k)$ be a reduced instance.
  \begin{enumerate}[(i)]
  \item\label{item:bigT1:no-sharing} Two big maximal cliques of $G$ share at most one vertex.
  \item\label{item:bigT1:max-clique}If  $(G,k)$ is a yes-instance, then a big maximal clique of $G$ remains a maximal clique after applying a solution to $(G,k)$.
  \end{enumerate}  
\end{lemma}
\begin{proof}
  Let $K_1$ and $K_2$ be two big maximal cliques of $G$.  Suppose first that some vertex $u \in K_1 \setminus K_2$ is adjacent to more than $2k+1$ vertices in $K_2$.  Since $K_2$ is a maximal clique, we can find $v\in K_2 \setminus K_1$ nonadjacent to $u$, but then Rule~\ref{rul:nonedge-sunflower} would be applicable (to $u v$).  Hence, every vertex in $K_1 \setminus K_2$ has at most $2k+1$ neighbors in $K_2$, which implies   $|K_1\cap K_2| \le 2k + 1$.  By assumption, $|K_1| \ge 3k+2$ and $|K_2| \ge 3k+2$.  For each vertex in $K_1 \setminus K_2$, we can find $k+1$ non-neighbors in $K_2 \setminus K_1$.  Therefore, we can greedily find $k + 1$ pairs of distinct vertices $\{x_1, y_1\}$, $\ldots$, $\{x_{k + 1}, y_{k+1}\}$ such that for all $1\le i\le k + 1$, (a) $x_i\in K_1 \setminus K_2$ and $y_i\in K_2 \setminus K_1$; and (b) $x_i y_i\not\in E(G)$.  Rule~\ref{rul:edge-sunflower} would be applicable (to any edge in $G[K_1 \cap K_2]$) if $|K_1 \cap K_2| \ge 2$.  Therefore, $|K_1 \cap K_2| \le 1$, and this concludes the proof for assertion (i).
  
  Let $E_\pm$ be a solution to $(G,k)$ and $G^* = G \triangle E_\pm$.  By Proposition~\ref{pro:untouchable-edge}(i), a big maximal clique $K$ in $G$ remains a clique in $G^*$.  Let $v\in V(G)\setminus K$ and let $u\in K\setminus N_G(v)$.  Since  Rule~\ref{rul:nonedge-sunflower} is not applicable to $uv$, there are at most $2k + 1$ neighbors of $v$ in $K$.  Since $|K| \geq 3k+2$, at least one vertex in $K$ remains nonadjacent to $v$ in $G^*$ because $|E_+|\le k$.  Therefore, $K$ is a maximal clique in $G^*$ as well.
\end{proof}

 It is well known that a graph $G$ is diamond-free if and only if every pair of adjacent vertices is contained in exactly one maximal clique of $G$. 
 This characterization turns out to be fundamental for the results in this paper.  We say that a maximal clique of $G$ is of \emph{type \textsc{i}} if it shares two or more vertices with some other maximal clique, and \emph{type \textsc{ii}} otherwise (its intersection with any other maximal clique is either $0$ or $1$ vertex).
We can then rephrase the first sentence of this paragraph as: A graph is diamond-free if and only if it has no maximal clique of type \textsc{i}.

We use \bone{G}, \sone{G}, and \two{G} to denote, respectively, the set of big maximal cliques of type \textsc{i}, the set of small maximal cliques of type \textsc{i}, and the set of maximal cliques of type \textsc{ii}, of $G$.  A maximal clique in $G$ is in precisely one of them.

\begin{proposition}\label{obs:diamond-vertices}
  A vertex of a graph $G$ is in a maximal clique of type \textsc{i} if and only if it is contained in an induced diamond in $G$.
\end{proposition}
\begin{proof}
  Let $u, v, x, y$ be four vertices inducing a diamond in $G$ with cross edge $u v$.  We can find a maximal clique $K_1$ containing $u, v, x$ and a maximal clique $K_2$ containing $u, v, y$.  They are different because $x \in K_1\setminus K_2$ and $y\in K_2\setminus K_1$, hence both of type \textsc{i}\@.

  We now consider the ``only if'' direction.  Let $K_1$ be a maximal clique of type \textsc{i}; by definition, there is another maximal clique $K_2$ such that $K_1 \cap K_2 \ge 2$.
  For any vertex $x \in K_1\setminus K_2$ and any vertex $u \in K_1 \cap K_2$, we can find another vertex $v \in K_1 \cap K_2$ different from $u$ and a vertex $y \in K_2 \setminus K_1$ not adjacent to $x$ (because $K_2$ is maximal).  Clearly, these four vertices induce a diamond with cross edge $u v$.
\end{proof}

The following two statements help us understand edges added by a minimum solution.
\begin{proposition}\label{lem:replacement}
  Let $G$ be a diamond-free graph, and let $U \subseteq V(G)$ such that every vertex in $V(G) \setminus U$ is adjacent to at most one vertex of $U$.  If $G[U] \triangle E_\pm$ is diamond-free for a set $E_+$ of non-edges in $G[U]$ and a set $E_-$ of edges in $G[U]$, then so is $G \triangle E_\pm$.
\end{proposition}
\begin{proof}
  Suppose for contradiction that $G^* = G \triangle E_\pm$ contains a diamond; let $D$ be a set of vertices inducing a diamond in $G^*$.  Since $G[D]$ is not a diamond, at least one (non-)edge of this diamond belongs to $E_\pm$, and is between vertices of $U$.  On the other hand, $G[U] \triangle E_\pm$ remains diamond-free, hence $D\not\subseteq U$.  Therefore, $|D\cap U|$ is either two or three, but then a vertex in $D\setminus U$ is adjacent to at least two vertices of $D \cap U$ in $G$, a contradiction.
\end{proof}


\begin{lemma}\label{lem:addition}
    Let $E_\pm$ be a minimum solution to a reduced yes-instance $(G,k)$. 
    Every vertex incident to some edge in $E_+$ is contained in some small maximal clique of type \textsc{i} in $G$. 
\end{lemma}
\begin{proof}
  Let $G^* = G \triangle E_\pm$, where $uv$ is an edge in $E_+$, and let $U$ be a maximal clique of $G^*$ containing $u, v$.  We argue first that $v$ is in some induced diamond in $G[U]$.

  Suppose for contradiction that $v$ participates in no diamond in $G[U]$.  
  Let $X = N_G(v)\cap U$.  The subgraph $G[X]$ is a disjoint union of cliques: An induced path of length two would make a diamond with $v$.  Let $\{A_1, \dots, A_p\}$ be those nontrivial cliques (containing more than one vertex) in $G[X]$; let $B$ be the other vertices of $X$; and let $C = U\setminus N_G[v]$.  Then  $\{A_1, \dots, A_p, B, C\}$ is a partition of the set $U \setminus \{v\}$.  Note that $p$ or $|B|$ may be $0$, but $|C| > 0$ because $u\in C$.
  To arrive at a contradiction, we will construct a solution $E'_\pm$ for $G[U]$ whose size is smaller than the number of non-edges in $G[U]$. Assume such an $E'_\pm$ exists and let $G'$ be the graph obtained from $G^*$ by replacing $G^*[U]$ with $G[U] \triangle E'_\pm$.  Since $U$ is a type-\textsc{ii} maximal  clique of $G^*$, for each $x \in V(G)\setminus U$ we have $|N_{G^*}(x)\cap U| \le 1$.  By Proposition~\ref{lem:replacement}, $G'$ is diamond-free.  This would however imply a strictly smaller solution than $E_\pm$, contradicting that $E_\pm$ is a minimum solution of $(G, k)$. Now we show how to construct $E'_\pm$.
	
  {Case 1, $|B| \ge |C|$.}   We set $E'_+ = \emptyset$ and $E'_-$ the set of edges in $G[C]$.  No edge in $E'_-$ is incident to $v$ or $N(v)$, and hence $N(v)\cap U$ is still a disjoint union of cliques in $G'$.  On the other hand, no vertex $x\in C$ is in any diamond in $G'[U]$ because $N_{G'}(x)\cap U$ is an independent set. Thus, $G'[U]$ is diamond-free.  Since $B$ is an independent set of $G$, and $v$ is nonadjacent to $C$, we have
  \[
    |E_+\cap U^2| \ge {|B|\choose 2} + |C| \ge {|C|\choose 2} + |C| > |E'_-| = |E'_+\cup E'_-|.
  \]
  
  {Case 2, $|B| < |C|$.}   We set $E'_+$ to be the set of non-edges in $G[B \cup C]$, and $E'_-$ the set of edges between $B \cup C$ and $U\setminus \left(B \cup C\right)$.  To see that $G'[U]$ is diamond-free, note that its maximal cliques are $B \cup C$ and $\{v\} \cup A_i$ for $1 \le i \le p$, whose intersection is either $\{v\}$ or empty.
We then calculate the cardinality of $E_+\cap U^2$, which comprises three parts, those among $B\cup C$, which is exactly $E'_+$, those between $C$ and $v$, and those between $C$ and $A_i$'s.   Since $v$ does not belong to any diamond in $G[U]$, each vertex in $C$ is adjacent to at most one vertex of $A_i, 1\le i\le p$.  In other words, for each $x\in C$ and each $1\le i\le p$, the number of non-edges between $x$ and $A_i$ is at least one.   Therefore
  \[
    |E_+\cap U^2|\ge |E'_+| + |C| + |C| \times p > |E'_+| + |B| + |C| \times p \ge |E'_+| + |E'_-|. 
  \]
	
  Now that $v$ is in some induced diamond in $G[U]$, we can find a maximal clique $K$ of $G$ containing three of its vertices including $v$.  Since $K\ne U$ and $|K\cap U|\ge 3$, it cannot induce a maximal clique of $G^*$. Hence by Lemma~\ref{lem:bigT1}(\ref{item:bigT1:max-clique}), it is small.  This concludes the proof of the lemma.
\end{proof}

After delimiting the ends of the edges added by a minimum solution, we then turn to the ends of those edges {deleted} by a minimum solution. The next lemma states that some maximal cliques in $G$ remain maximal cliques after applying the solution, that is, none of the edges inside those cliques are deleted. 

%

\begin{lemma}\label{lem:type-2}
  Let $E_\pm$ be a minimum solution to an instance $(G, k)$, and let $K$ be a maximal clique of type \textsc{ii} in $G$.  
  If  $E_+$ contains neither (i) an edge between $u \in K$ and $v \in N(K)$, nor (ii) two edges between vertices of $K$ and the same vertex in $V(G)\setminus K$, then $K$ remains a maximal clique (of type \textsc{ii}) in $G\triangle E_\pm$.
\end{lemma}
\begin{proof}
  Let $G^* = G \triangle E_\pm$.	Since $K$ is a type-\textsc{ii} maximal clique of $G$, each vertex $v\in V(G) \setminus K$ has at most one neighbor in $K$.  By the assumption that $E_+$ contains neither (i) nor (ii), this remains true in $G + E_+$ and $G^*$.  
  On the other hand, $E_-$ cannot contain edges of $G[K]$; otherwise, by Proposition~\ref{lem:replacement}, $G^*$ remains diamond-free after replacing $G^*[K]$ by $G[K]$, which implies a strictly smaller solution than $E_\pm$.  Therefore, $K$ is a maximal clique in $G^*$.
\end{proof}

The next corollary follows from Lemma~\ref{lem:addition} and Lemma~\ref{lem:type-2}.
\begin{corollary}\label{lem:deletion}
  Let $E_\pm$ be a minimum solution to a reduced yes-instance $(G, k)$, and let $K$ be a maximal clique of $G$ containing both ends of an edge in $E_-$.  Then either $K \in \sone{G}$, or $K \in \two{G}$ and $K$ intersects one clique in \sone{G}.
\end{corollary}

\begin{figure}[h]
\centering
\begin{tikzpicture}[scale=2]\small
  \begin{scope}[every node/.style={myv}, every path/.style={thin}]
    \graph[simple] { 
      subgraph K_n [n=14, empty nodes, nodes={rotate=-540/14}, clockwise,radius=2.5cm] ;
    };
  \end{scope}
  \begin{scope}[every node/.style={myv}]
    \node ["$v_7$" right, right = 10mm of 3, yshift=3mm] (v7) {};
    \node ["$v_8$" right, right = 10mm of 3, yshift=-3mm] (v8) {};
    \node ["$v_9$" right, below = 4mm of v8] (v9) {};    
\draw (3) edge (v7);
\draw (3) edge (v8);
\draw (v7) edge (v8);
\draw (3) edge (v9);
\draw (4) edge (v9);

\node ["$v_1$", above left = 6mm of 10] (v1) {};
\node ["$v_0$", left  = 5mm of v1] (v0) {};
\node ["$v_2$", left = 14mm of 10] (v2) {};
\draw (10) -- (v2) -- (v0) -- (10) -- (v1) -- (v2) (v0) -- (v1);

\node ["$v_3$" left, below = 8mm of v2] (v3) {};
\node ["$v_4$" above right, right = 5mm of v3, yshift=1mm] (v4) {};
\node ["$v_6$" below, below = 6mm of v4, xshift=2mm] (v6) {};
\node ["$v_5$" below, left = 9mm of v6] (v5) {};
\draw (v3) -- (v4) -- (v5) -- (v6) -- (v3) -- (v5) (v4) -- (v6);
\foreach \i in {3,...,6}{
  \draw (v\i) edge (9);
  \draw (v\i) edge (v2);
}
\end{scope}
\begin{scope}[blue]
  \node["$u_1$" below, xshift=-1.3mm, at={(10)}] {};
  \node["$u_2$" below, at={(9)}] {};    
  \node["$u_3$" below, xshift=1.3mm, at={(3)}] {};
  \node["$u_4$" below, xshift=1mm, at={(4)}] {};    
\end{scope}
\end{tikzpicture}
\caption{An example with $k = 4$, of which a minimum solution is $\{+u_2 v_2, -u_1 v_2, -v_0 v_1, -u_3 v_{9}\}$.  (Note that $u_1 v_2$ and $v_0 v_1$ are not in any diamond of $G$.)  It has six maximal cliques, $K_1 = \{v_0, v_1, v_2, u_1\}$, $K_2 = \{v_2, v_3, v_4, v_5, v_6\}$, $K_3 = \{u_2, v_3, v_4, v_5, v_6\}$, $K_4 = \{u_3, v_7, v_8\}$, $K_5 = \{u_3, u_4, v_{9}\}$, while $K_6$ comprises of $u_1, u_2, u_3, u_4$ and other ten unlabeled vertices.  Four of these maximal cliques, $K_2$, $K_3$, $K_5$, and $K_6$, are of type \textsc{i}, of which only $K_6$ is big, the other two of type \textsc{ii}\@.  All 14 labeled vertices are vulnerable, and the other 8 unlabeled vertices are guarded. }
\label{fig:example}
\end{figure}
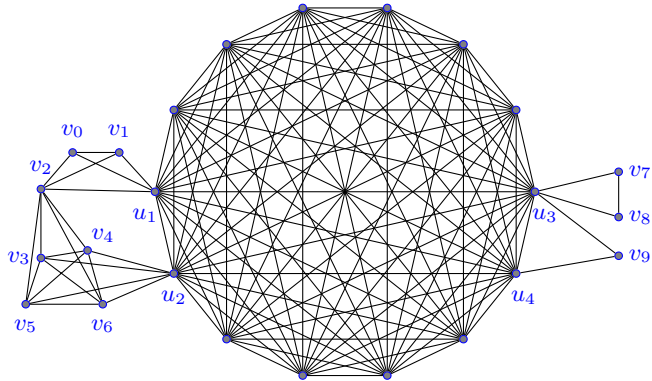

Lemma~\ref{lem:addition} and Corollary~\ref{lem:deletion} motivate the following definition.  We say that a vertex $v$ is {\em vulnerable} in graph $G$ if (1) there exists some $K\in \sone{G}$ containing $v$; or (2) there are intersecting maximal cliques $K_1\in \sone{G}$ and $K_2\in \two{G}$ such that $v\in K_2$.  A vertex is \emph{guarded} if it is  not vulnerable.  Lemma~\ref{lem:addition} and Corollary~\ref{lem:deletion} can be summarized as: No (non-)edge in a minimum solution can be incident to a guarded vertex.  See Figure~\ref{fig:example} for an illustration.

\section{The kernel}
\label{sec:kernel}

We partition the vertex set of a reduced graph into five parts, and deal with them separately.
\begin{enumerate}[(i)]
\item\label{item:1} vertices in small maximal cliques of type \textsc{i} (all of them are vulnerable); 
\item\label{item:2} vulnerable vertices in big maximal cliques of type \textsc{i} but not in the previous part;
\item\label{item:3} other vulnerable vertices (not in any maximal cliques of type \textsc{i});
\item\label{item:4} guarded vertices in (big) maximal cliques of type \textsc{i}; and
\item\label{item:5} other guarded vertices (not in any maximal cliques of type \textsc{i}).
\end{enumerate}
Note that for this purpose we do \textit{not} need to enumerate the maximal cliques.  The key observation is that we can easily find the cross edges of all diamonds by enumeration, from which we can identify all vertices and edges in maximal cliques of type \textsc{i}.  We use the procedure \texttt{partition} presented in Figure~\ref{fig:alg-partition}, which computes this partition in three steps: It first finds all vertices in a maximal clique of type \textsc{i}, from which it identifies those in a small maximal clique of type \textsc{i}, and finally it uses them to get all vulnerable vertices.


\begin{figure}[h!]
  \centering
  \tikz\path (0,0) node[draw, text width=.9\textwidth, rectangle, rounded corners, inner xsep=20pt, inner ysep=10pt]{
    \begin{minipage}[t!]{\textwidth} \small
      {\sc Input}: a reduced instance ($G, k$).
      \\
      {\sc Output}: vertices in the five parts have (i) mark ``small,'' (ii) marks ``vulnerable'' and ``type \textsc{i},''  (iii) mark ``vulnerable,'' (iv) mark ``type \textsc{i},'' and (v) no mark, respectively.
			
      \begin{tabbing}
        Aa\=Aaa\=aaa\=Aaa\=MMMMMMAAAAAAAAAAAAA\=A \kill
        1.\> {\bf for} each edge $u v\in E(G)$ where $N(u)\cap N(v)$ does not induce a clique {\bf do}
        \\
        1.1.\>\> mark $u v$ ``cross edge'';
        \\
        1.2.\>\> mark $u, v$ and all vertices in $N(u) \cap N(v)$ as ``type \textsc{i}'';
        \\
        1.3.\>\> mark all edges between these vertices as ``type \textsc{i}'';
        \\
        \> {\textbackslash\!\!\textbackslash {\em a vertex is in a maximal clique of type \textsc{i} if and only if it's marked ``type \textsc{i}.''}}
        \\
        2.\> {\bf for} each marked vertex $v$ {\bf do}
        \\
        2.1.\>\> {\bf if} $N(v)$ does not induce a cluster (a disjoint union of cliques) {\bf do} mark $v$ as ``small'';
        \\
        2.2.\>\> {\bf else if} a clique in $N(v)$ of size $\le 3 k$ contains a cross edge {\bf do} mark $v$ as ``small'';
        \\
        3.\> {\bf for} each unmarked edge $u v\in E(G)$ {\bf do}
        \\
        3.1.\>\> find the maximal clique $K$ containing $u$ and $v$;
        \\
        3.2.\>\> {\bf if} $K$ contains any vertex marked ``small'' {\bf then} mark all vertices in $K$ ``vulnerable'';
        \\    
        3.3.\>\> mark every edge in $K$ ``checked.''
      \end{tabbing}
		\end{minipage}
	};
	\caption{The procedure  \texttt{partition}.}
	\label{fig:alg-partition}
\end{figure}

It is easy to check that procedure \texttt{partition} runs in polynomial time.  We now show its correctness.

\begin{lemma}\label{lem:partition}
  Procedure {partition} is correct.
\end{lemma}
\begin{proof}
  We start from proving a simple property:
  \begin{quote}
  A maximal clique $K$ is of type \textsc{i} if and only if it contains\\ both ends of the cross edge of a diamond.\hfill ($\natural$)
  \end{quote}
  
  Consider first the ``only if'' direction.  By definition, there exists another maximal clique $K'$ such that $|K\cap K'|\ge 2$.  We can find two vertices $u, v\in K\cap K'$; while by the maximality of $K$ and $K'$, we can find $x\in K\setminus K'$ and $y\in K'\setminus K$.  These four vertices induce a diamond with cross edge $u v$.  Consider then the ``if'' direction.  Let $u v$ be a cross edge of a diamond such that $u, v\in K$, and let $x,y$ be the other vertices of the diamond.  We can find two maximal cliques containing $u, v, x$ and $u, v, y$ respectively.  They are different and hence at least one is different from $K$.  Therefore, $K$ is of type \textsc{i}\@.
	
  An edge $uv \in E(G)$ is a cross edge if and only if $N(u)\cap N(v)$ does not induce a clique; this justifies step~1.1.  Steps~1.2 and 1.3 follow from property ($\natural$).
	
  Step~2 considers all vertices in maximal cliques of type \textsc{i}.  If some component of $G[N(v)]$ is not a clique, we can find a path $x y z$ of length two.  There are two different maximal cliques containing $v, x, y$ and $v, y, z$ respectively.  Both are of type \textsc{i}, and hence by Lemma~\ref{lem:bigT1}(\ref{item:bigT1:no-sharing}), at least one of them is small.  Step~2.2 also follows from property ($\natural$).  If a vertex $v$ is not marked in step~2, then every maximal clique containing $v$ is either big or of type \textsc{ii}.  Therefore, all vertices in small maximal cliques of type \textsc{i} have been correctly identified in step~2.  
	
  Step~3 finds other vulnerable vertices.  By definition, such a vertex is in some maximal clique of type \textsc{ii}\@.  If a type-\textsc{ii} maximal clique consists of an isolated vertex, it is guarded and not marked in step~3.  We may hence consider only nontrivial maximal cliques.  All edges in a type-\textsc{ii} maximal clique remain unmarked.  Note that any two vertices of a type-\textsc{ii} maximal clique determines this clique: It is the only maximal clique that contains these two vertices.  Vertices in the clique are vulnerable if and only if it contains a vertex marked ``small.''  We only need to check the clique $K$ once, so we mark them to avoid unnecessary repetition in step~3.3.  After step~3, all type-\textsc{ii} maximal cliques have been checked, and the algorithm is complete.
\end{proof}


\subsection{Maximal cliques of type {\sc i} }
We start from the vertices that are in some small type-\textsc{i} maximal cliques of $G$; let them be denoted by  $S(G)$, i.e., $S(G) = \bigcup\limits_{K\in \sone{G}} K$.  Noting that the final graph has no small type-\textsc{i} maximal cliques, we can bound the size of $S(G)$ by relating vertices in it to edges in a minimum solution. 

\begin{lemma}\label{lem:smallT1-vertex-count}
  If  $(G,k)$ is a reduced yes-instance, then $|S(G)| \le 18 k^3 + 2k$.
\end{lemma}
\begin{proof} 
  Let $E_{\pm}$ be a minimum solution of $(G, k)$.  Let $X = \bigcup_{x y\in E_\pm} \{x, y\}$ and $Y = \bigcup_{x y\in E_\pm} N_G(x)\cap N_G(y)$, i.e., all vertices incident to a (non-)edge in the solution and respectively, all vertices that is a common neighbor of the two ends of a (non-)edge in the solution.  Note that $|X| \le 2 k$, and by Proposition~\ref{pro:edited-neighborhood}, $|Y| \le 3k\cdot |E_{\pm}| \le 3k^2$.
  Since
  \[
    |S(G)\cap (X\cup Y)| \le |X\cup Y|  \le 3k^2 + 2k,
  \]
  it suffices to bound $S(G)\setminus (X\cup Y)$.  A vertex $v\in S(G)\setminus (X\cup Y)$ cannot be contained in two type-{\sc i} maximal cliques of $G$ if they share more than one vertex: Otherwise, there is a diamond (as in Proposition~\ref{obs:diamond-vertices}) in $N_G[v]$, but then $v$ has to be in $X\cup Y$, a contradiction. 

  Let us now consider the set of small type-{\sc i} maximal cliques of $G$ that contain vertices from $S(G)\setminus (X\cup Y)$, which we denote by $\cal K'$.  We argue by contradiction that any pair of cliques in $\cal K'$ shares at most one vertex.  
  Suppose otherwise, there are two maximal cliques $K_1, K_2\in \cal K'$ with $|K_1\cap K_2| \ge 2$.  We have seen that $K_1\cap K_2$ is disjoint from $S(G)\setminus (X\cup Y)$.
  Now let $u\in K_1\setminus K_2$ and $v\in K_2\setminus K_1$ be two vertices in $S(G)\setminus (X\cup Y)$.  Then there is a diamond with $u$ and two vertices in $K_1\cap K_2$ and one vertex in $K_2\setminus K_1$.  But by the assumption $u\not\in X\cup Y$, we cannot add or delete any edge incident to $u$; on the other hand, $v\not\in X\cup Y$ forbids the deletion of other three edges, a contradiction.
  
  Let $v\in S(G)\setminus (X\cup Y)$, and let $K$ be a clique in $\cal K'$ containing $v$.  By definition, there exists a diamond in which (1) $v$ is a degree-two vertex; (2) the two degree-three vertices are in $K$; and (3) the other degree-two vertex is not in $K$.  Since $v$ is not in $X \cap Y$, one of the two edges of this diamond that are incident to the other degree-two vertex has to be in $E_-$.  In other words, $K$ contains for some edge $x y\in E_-$, one in $\{x, y\}$ and a common neighbor of $x, y$.  By Proposition~\ref{pro:edited-neighborhood}, for each edge $x y\in E_-$, there are at most $3 k$ vertices in $N_G(x)\cap N_G(y)$;  for each $z\in N_G(x)\cap N_G(y)$, there can be at most one clique in $\cal K'$ containing $x, z$ and  at most one clique in $\cal K'$ containing $y, z$.  Therefore, there can be at most $3 k \cdot 2 \cdot |E_-| \le 6 k^2$ cliques in $\cal K'$.  By definition, each clique in it is small and has at most $3k + 1$ vertices, of which at least two are not in $S(G)\setminus (X\cup Y)$.  Hence
  \[
    |S(G)\setminus (X\cup Y)| \le (3k - 1) \cdot 6 k^2 = 18 k^3 - 6 k^2.
  \]
  Putting the two parts together, we have $|S(G)| \le 18 k^3 + 2k$.
\end{proof}

The next are the big maximal cliques of type \textsc{i}\@. By Lemma~\ref{lem:bigT1}(\ref{item:bigT1:max-clique}), a clique in \bone{G} remains a maximal clique after a minimum solution is applied to $G$.  We bound first the number of big type-{\sc i} maximal cliques. 

\begin{lemma}\label{lem:bigT1-count}
  If $(G,k)$ is a reduced yes-instance, then $|\bone{G}| \le 6 k^2$.
\end{lemma}
\begin{proof}
By Lemma~\ref{lem:bigT1}, the only way to transform a big maximal clique of type \textsc{i} into one of type \textsc{ii} is deleting edges incident to it. For an edge $e = uv \in E_-$, denote by ${\cal K}_e$ the set of big type-\textsc{i} maximal cliques containing one in $\{u, v\}$, and one vertex in $N(u) \cap N(v)$. Note that ${\cal K}_b(G) = \bigcup_{e\in E_-}{\cal K}_e$. By Proposition~\ref{pro:edited-neighborhood}, ${\cal K}_e$ has at most $6k$ maximal cliques. Then $|{\cal K}_b(G)| \le 6k \cdot |E_-| = 6k^2$. 
\end{proof}


To bound the number of vertices in big type-{\sc i} maximal cliques, it suffices to bound their sizes, for which we introduce another reduction rule. 

\begin{redrule}\label{rul:big-t1}
  Let $K\in \bone{G}$ with $|K|\ge 3k + 3$.  If $K$ contains a guarded vertex $x$ that does not occur in any other type-\textsc{i} maximal clique of $G$, delete it.
\end{redrule}

\begin{lemma}\label{lem:rul:big-t1}
  Rule \ref{rul:big-t1} is safe: A reduced instance $(G, k)$ is a yes-instance if and only if $(G - x, k)$ is a yes-instance.
\end{lemma}
\begin{proof}
  It is easy to see that $(G-x, k)$ is a reduced instance, and every solution of $(G,k)$ confined to $G-x$ is a solution of $(G-x, k)$. 
  For the ``if'' direction, let $E_\pm$ be a minimum solution of $(G-x, k)$, and let $G^* = G\triangle E_\pm$.  Note that $(G-x)\triangle E_\pm= G^*-x$, and it is diamond-free.  No edge in $E_\pm$ is incident to $x$, and hence $N_G(x) = N_{G^*}(x)$, which we simply denote by $N(x)$. By Proposition~\ref{obs:diamond-vertices}, it suffices to prove that each maximal clique of $G^*$ containing $x$ is of type \textsc{ii}.  For this purpose, we show that each component of $G[N(x)]$ is either a single vertex or a type-\textsc{ii} maximal clique in $G^* - x$.

  Note that $K\setminus \{x\}$ is a big maximal clique in $G-x$: It is a clique of size at least $3k + 2$, and its maximality follows from Lemma~\ref{lem:bigT1}(i).  Hence, by Lemma~\ref{lem:bigT1}(\ref{item:bigT1:max-clique}), $K\setminus \{x\}$ is a maximal clique (of type \textsc{ii}) in $G^*-x$. Since $x$ is a guarded vertex that does not occur in any other type-\textsc{i} maximal clique, every other maximal clique $K'$ containing $x$ in $G$ is of type \textsc{ii}, and it cannot intersect any small type-\textsc{i} maximal clique. Therefore, by Lemma~\ref{lem:addition}, no edge added by $E_+$ can be incident to any vertex in $N(x)$. 
  From Lemma~\ref{lem:type-2} we can conclude that $K' \setminus \{x\}$ either contains only a vertex or is a maximal clique (of type \textsc{ii}) in $G^* - x$. 
  
  Since no edge added by $E_+$ is between two vertices in $N(x)$ and since $x$ is a guarded vertex, each component of $G[N(x)]$ is either $K\setminus \{x\}$ or $K' \setminus \{x\}$, hence is either a single vertex or a type-\textsc{ii} maximal clique in $G^* - x$.  This concludes the proof.
\end{proof}

\begin{lemma}\label{lem:big-t1}
  Let $(G, k)$ be a reduced yes-instance.  If Rule~\ref{rul:big-t1} is not applicable, then for each $K\in \bone{G}$, we have that $|K| = O(k^3)$.
\end{lemma}
\begin{proof}
  Without loss of generality, assume that $|K| \ge 3k + 3$. 
  Since Rule~\ref{rul:big-t1} is not applicable, $K$ does not contain any guarded vertex shared by other big type-\textsc{i} maximal cliques.  Thus, every vertex in $K$ is either a vulnerable vertex, or a guarded vertex in more than one big \tone\ maximal clique.
  Let $U_1$ and $U_2$ be the set of vulnerable vertices in $K \cap S(G)$ and $K \setminus S(G)$ respectively. By the definition, each vertex in $U_2$ is adjacent to some vertex in $S(G) \setminus U_1$ by an edge of \ttwo\ maximal clique. For each vertex $v \in S(G) \setminus U_1$, the cardinality of $U_2 \cap N(v)$ is at most one; otherwise, there is a type-\textsc{i} maximal clique containing $U_2 \cap N(v)$ and $v$ which by Lemma~\ref{lem:bigT1}(\ref{item:bigT1:no-sharing}) is small, contradicting to $U_2 \subseteq K \setminus S(G)$. 
  Therefore, $|U_2| \le |S(G) \setminus U_1|$, and by Lemma \ref{lem:smallT1-vertex-count}, $K$ contains at most $18k^3 + 2k$ vulnerable vertices.
  By Lemma~\ref{lem:bigT1}(\ref{item:bigT1:no-sharing}), every pair of big \tone\ maximal cliques shares at most one vertex.  Hence, by Lemma~\ref{lem:bigT1-count}, 
  $K$ contains at most $6k^2$ guarded vertices that appear in some other big maximal cliques of type \textsc{i}.  
  Putting them together we get $|K|\leq 18k^3+2k+6k^2$.
\end{proof}

The next corollary follows immediately from Lemmas~\ref{lem:bigT1-count} and~\ref{lem:big-t1}. 
\begin{corollary}\label{col:bigT1-vertex-count}
  Let $(G, k)$ be a reduced yes-instance.  If Rule~\ref{rul:big-t1} is not applicable, then the number of vertices that are contained in some cliques in \bone{G} is $O(k^5)$.
\end{corollary}
\subsection{Maximal cliques of type {\sc ii} }

We have bounded the number of vertices in all maximal cliques of type \textsc{i}, and it remains to bound the number of vertices that occur \textit{only} in maximal cliques of type \textsc{ii}.  Let $T(G)$ denote these vertices, i.e., $T(G) = V(G)\setminus \bigcup\limits_{K\in \sone{G}\cup \bone{G}} K$.   It may not be surprising that we can delete all the guarded vertices in them.

\begin{redrule}\label{rul:guarded-vertices}
	If there is a guarded vertex $x$ not in any type-\textsc{i} maximal clique of $G$, delete it.
\end{redrule}

\begin{lemma}\label{lem:guarded-vertices}
  Rule \ref{rul:guarded-vertices} is safe: A reduced instance $(G, k)$ is a yes-instance if and only if $(G - x, k)$ is a yes-instance.
\end{lemma}
\begin{proof}
  It is easy to see that $(G-x, k)$ is a reduced instance, and every solution of $(G,k)$ confined to $G-x$ is a solution of $(G-x, k)$. 
  For the other direction, let $E_\pm$ be a minimum solution of $(G-x,k)$, and it is sufficient to show that $x$ is not part of any diamond in $G^* = G\triangle E_\pm$.  
  Note that $x$ is a vertex which is part of only \ttwo\ maximal cliques in $G$
  and not adjacent to any vertex in small \tone\ maximal cliques in $G$. Therefore, by Lemma~\ref{lem:addition}, none of the vertices in $N(x)$ is incident to any edges of $E_+$. If $x$ is part of a diamond in $G^*$, then it is formed by a deletion of an edge in $G[N[x]]$ by $E_-$. But this is not possible by Corollary~\ref{lem:deletion},
  as none of the edges in $G[N[x]]$ is part of any \ttwo\ maximal clique which intersects with a small \tone\ maximal clique in $G-x$.
\end{proof}

If Rule~\ref{rul:guarded-vertices} is not applicable, then all vertices in $T(G)$ are vulnerable.
As demonstrated in Figure~\ref{fig:example}, an edge may be deleted from a maximal clique of type \textsc{ii}.  In that example, neither end of the deleted edge $v_0 v_1$ is in any maximal clique of type \textsc{i}.  This can happen only \textit{after} some modification happens in the neighborhood of this vertex---$u_2 v_2$ added in the example.  Indeed, we may consider the added/deleted edges stepwise, then there is an order such that each edge is added/deleted \textit{only if} it is in some diamond.  One modification may introduce new diamond(s) not in the original graph.  For example, neither $v_0 v_1$ nor $u_1v_2$ is in a diamond of $G$, but the addition of $u_2 v_2$ jeopardizes $u_1v_2$, whose deletion consequently brings $v_0 v_1$ down.

This example is actually exemplary: The only way to bring an edge in a maximal clique $K$ of type \textsc{ii} to a diamond is through adding edge(s) between $K$ and other vertices.  
According to Proposition~\ref{pro:untouchable-edge}, however, this would not happen when $|K| \ge k + 3$.  In other words, to make sure a large clique in \two{G} is immutable to future modifications, it suffices to keep $k + 3$ of its vertices. 
This motivates the following reduction rule, whose statement is however more complex than  previous ones.  The main trouble here is that we are not allowed to delete all but $k + 3$ guarded vertices from a clique in \two{G}, because it may be required for another clique in \two{G}.

%

For a pair of vertices $u, v$, we denote by $N(u, v)$ the set of common neighbors of $u$ and $v$ not in $S(G)$, i.e., $N(u, v) = (N(u)\cap N(v)) \setminus S(G)$.

\begin{proposition}\label{pro:hanging-vertices}
  Let $u,v$ be two vertices in $G$.  If $uv \not\in E(G)$, then $N(u, v)$ form an independent set.  Moreover, if $uv\in E_+$ for a solution $E_\pm$ of $(G,k)$, then $|N(u, v)|\leq k$.
\end{proposition}
\begin{proof}
	If $G[N(u,v)]$ has an edge, say $xy$, then $\{u,v,x,y\}$ forms a diamond. 
	There are two type-\textsc{i} maximal cliques containing $\{x,y,u\}$ and $\{x,y,v\}$ respectively. By Lemma~\ref{lem:bigT1}(\ref{item:bigT1:no-sharing}), at least one of them is small, contradicting to $x,y \notin S(G)$. The second claim follows from Proposition~\ref{pro:untouchable-edge}.
\end{proof}

Our last rule would keep at most $k + 1$ from such sets.  To avoid unnecessary clutters, we simply say we mark $k + 1$ vertices in $N(u, v)$, even if its size is smaller than $k + 1$; in which case, we mark all of them.

\begin{redrule}\label{rul:vulnerable-vertices}
  For each pair of vertices $u, v \in S(G)$, 
  arbitrarily mark $k + 1$ vertices in $N(u, v)$.  If $|N(u, v)|\le k$, then for each vertex $w \in N(u, v)$, arbitrarily mark $k+1$ vertices in $N(u, w)$ and $k+1$ vertices in $N(v, w)$.  If there is an unmarked vertex $x$ in $T(G)$, delete it.
\end{redrule}

\begin{lemma}\label{lem:vulnerable-vertices}
  Rule \ref{rul:vulnerable-vertices} is safe: A reduced instance $(G, k)$ is a yes-instance if and only if $(G - x, k)$ is a yes-instance. 
\end{lemma}
\begin{proof}
  It is easy to see that $(G-x, k)$ is a reduced instance, and every solution of $(G,k)$ confined to $G-x$ is a solution of $(G-x, k)$. 
  For the ``if'' direction, let $E_\pm$ be a minimum solution of $(G-x, k)$, and let $G^* = G\triangle E_\pm$.  We show that each maximal clique of $G^*$ containing $x$ is a maximal clique of $G$ and is of type \textsc{ii} in $G^*$.  Since Proposition~\ref{obs:diamond-vertices} implies that deleting a vertex not in any \tone\ maximal clique does not alter type-\textsc{i} maximal cliques, we have $S(G') = S(G)$.
	
  Let $K$ be a maximal clique of $G$ containing $x$; note that $K$ is a maximal clique of type \textsc{ii} in $G$, as $x \in T(G)$.  We argue that $|N_{G^*}(y)\cap K| \le 1$ for every $y\in V(G) \setminus K$.
  Since $K$ is a maximal clique of type \textsc{ii} in $G$, we have (1) $|N_{G}(y)\cap K|$ is either $0$ or $1$; and (2) for every pair of vertices $u, v\in K$, 
  \[
    N(u, v)\subseteq N_G(u)\cap N_G(v) = K.
  \]

  Suppose first that there are at least two edges between $y$ and $K$ in $E_+$.
  Let $u,v \in K$ be two vertices such that $yu, yv \in E_+$.  Then by Lemma~\ref{lem:addition}, $u,v \in S(G')$, and hence $u, v \in S(G)$.  Clearly, 
  $x\neq u$, $x\neq v$ and $x$ is an unmarked vertex in $N(u, v)$. Further, there are $k+1$ marked vertices in $N(u, v)$.  It follows that $|K\setminus \{x\}|\ge k + 3$, and $E_-$ does not have
  any edge in $G'[K\setminus \{x\}]$ by Proposition~\ref{pro:untouchable-edge}(i).  
  Therefore, for each marked vertex $z \in N(u, v)$ that is not adjacent to $y$, the set $\{u,v, y,z\}$ induces a diamond in $G' + \{yu, yv\}$.  The only edge we can edit is $yz$, but $|N_G(y)\cap K| \le 1$, and there are at least $k + 2$ edges between $y$ and $K$, which is impossible.

  Hence, at most one edge can be added between $y$ and $K$ by $E_+$.  If $|N_{G}(y)\cap K| = 0$, or $|N_{G}(y)\cap K| = 1$ but the only edge between $y$ and $K$ is deleted, then it is trivial that
  $y$ is adjacent to at most one vertex of $K$ in $G^*$.  Suppose that $N_{G^*}(y)\cap K = \{u, v\}$ while only $u$ is in $N_G(y)$; note that $y u \not\in E_-$ and $yv \in E_+$.  By Lemma~\ref{lem:addition}, $y, v \in S(G')$, and hence in $S(G)$.  
  According to Proposition~\ref{pro:hanging-vertices}, there are at most $k$ vertices in $N(v, y)$ in $G'$.  If $u \notin S(G)$, then it has been marked; hence  $x \ne u$.  Also, $x \neq v$ as $x \in T(G)$.  By the rule, no matter whether $u$ is in $S(G)$ or not, we should have marked vertices in $N(u, v)$. Since $x\in N(u,v)$ but is not marked, we have $|N(u, v)| > k + 1$.  Let $z$ be any marked vertex in $N(u, v)$; it is not in $N_G(y)$ by assumption.
But then $\{u,v, y,z\}$ induces a diamond in $G' + yv$, in which we have to add the missing edge $y z$, which requires $|E_+| > k$, a contradiction.
	
 We have thus concluded  $|N_{G^*}(y)\cap K| \le 1$ for each vertex $y$ in $V(G) \setminus K$.  By Proposition~\ref{lem:replacement}, $K\setminus \{x\}$ remains a clique in $G^* - x$, otherwise we can find a strictly smaller solution. Then $K$ is a maximal clique of type \textsc{ii} in $G^*$.   On the other hand, according to Proposition~\ref{pro:hanging-vertices}, no edge is added between two vertices of $N_G(x)$.  Therefore, $N(x)$ induces exactly the same subgraph in $G$ and $G^*$.  Hence, any maximal clique of $G^*$ containing $x$ is a maximal clique of $G$ as well, hence of type \textsc{ii} in $G^*$.  This concludes the proof of the lemma.         
\end{proof}

Now Theorem~\ref{thm:main} follows by counting numbers of different kinds of vertices.

\begin{proof}[Proof of Theorem~\ref{thm:main}]
We show first that Rules~\ref{rul:big-t1}--\ref{rul:vulnerable-vertices} can be applied in polynomial time. 
For a guarded vertex $x$, $N(x)$ induces a cluster graph and each maximal clique in the cluster graph together with $x$ forms the maximal cliques of $G$ containing $x$. 
Recall that a maximal clique is of type \textsc{i} if and only if it contains both ends of a cross edge. 
Since the procedure \texttt{partition} finds all guarded vertices (no mark) and cross edges, we can find for each guarded vertex all type-\textsc{i} maximal cliques and type-\textsc{ii} maximal cliques containing it in polynomial time. 
Therefore, both Rules~\ref{rul:big-t1} and Rule~\ref{rul:guarded-vertices} can be applied in polynomial time. 
Moreover, the procedure \texttt{partition} finds all vertices in $S(G)$ (mark ``small'') and $T(G)$ (no mark ``type \textsc{i}''), and hence Rule~\ref{rul:vulnerable-vertices} can be applied in polynomial time.

We claim that if none of Rules~\ref{rul:big-t1}--\ref{rul:vulnerable-vertices} is applicable to a reduced yes-instance $(G,k)$, then $|V(G)| = O(k^8)$. 	
By Lemma \ref{lem:smallT1-vertex-count}, the number of vertices in small type-{\sc i} maximal cliques is $|S(G)| = O(k^3)$. 
By Corollary~\ref{col:bigT1-vertex-count}, we have $O(k^5)$ vertices in big type-{\sc i} maximal cliques. 
For each pair of vertices $u, v$ in $S(G)$, we mark at most $k+1$ common neighbors of them. 
And for each common neighbor $w$ of $u, v$, we mark at most $2k+2$ vertices: $k+1$ vertices in $N(u, w)$ and $k+1$ vertices in $N(v, w)$. 
Hence $|T(G)| = O(k^8)$, and $|V(G)| = O(k^3) + O(k^5) + O(k^8) = O(k^8)$.
\end{proof}

\section{A cubic kernel for diamond-free edge deletion}
\label{sec:deletion}

We now present a cubic-vertex kernel for the diamond-free edge deletion problem.  Note that if $G -E_-$ is diamond-free, then $E_-$ can be viewed as a solution to the diamond-free editing problem as well, where $E_+=\emptyset$.  Therefore, most statements, except those on minimum solutions, also hold for $E_-$.  
We will need Rule~\ref{rul:edge-sunflower} from page~\pageref{rul:edge-sunflower}; for the sake of completeness, we include it here.

\setcounter{redrule}{0}
\begin{redrule}
  \label{rul:DD-1}
  If there exist an edge $uv$ and $2k + 2$ distinct vertices $x_1, y_1, \ldots, x_{k + 1}, y_{k + 1}$ in $N(u) \cap N(v)$ such that $x_i y_i\not\in E(G)$ for all $1\le i \le k + 1$, then delete $uv$ and decrease $k$ by one.
\end{redrule}

The correctness of the following rule is also straightforward. 
\begin{redrule}
  \label{rul:DD-2}
  Mark an edge $uv$ ``permanent'' if there are $2k + 2$ distinct vertices $x_1, y_1, \ldots, x_{k + 1}, y_{k + 1}$ in $N(u) \cap N(v)$ such that $x_i y_i\in E(G)$ for all $1\le i \le k + 1$.  If there exists a diamond consisting of only permanent edges, return a trivial no-instance.
\end{redrule}

An instance of the diamond-free edge deletion~problem is \textit{reduced} if neither of Rules~\ref{rul:DD-1} and \ref{rul:DD-2} is applicable.  Henceforth we are concerned exclusively with reduced instances.  
\begin{proposition}
	\label{lem:bigT1:deletion}
	Two big maximal cliques of a reduced instance $(G,k)$ share at most one vertex.
\end{proposition}
\begin{proof}
	Each edge in a maximal clique would be marked permanent by Rule~\ref{rul:DD-2}.  Therefore, if two big maximal cliques share more than one vertex, there is a diamond in them, consisting of only permanent edges.
\end{proof}

Again, given any minimal solution $E_-$, we can view the edges as deleted in a sequence, such each edge is in a diamond when it is deleted.  According to Lemma~\ref{lem:type-2}, a type-\textsc{ii} maximal clique would remain so during the course.  If a vertex is not in any type-\textsc{i} maximal clique of $G$, then by Proposition~\ref{obs:diamond-vertices}, no edge incident to it will be deleted.  These vertices and edges are thus \textit{irrelevant} to the problem; we can actually delete them from the graph.
\begin{redrule}\label{rul:DD-not-part}
  Delete all edges and vertices not in any maximal clique of type {\sc i}. 
\end{redrule}
\begin{lemma}
	Rule~\ref{rul:DD-not-part} is safe. 
\end{lemma}
\begin{proof}
	Let $G'$ be the graph obtained by deleting all edges in type-{\sc ii} maximal cliques.  Note that a vertex is not in any type-{\sc i}  maximal clique of $G$ if and only if it is isolated in $G'$.  Therefore, to show the safeness of Rule~\ref{rul:DD-not-part}, it suffices to prove that $(G,k)$ is a yes-instance if and only if $(G',k)$ is a yes-instance. 
	
	Let $E_-$ be a minimum solution of $(G, k)$, and let $E'_- = E_-\cap E(G')$.  We claim that $G' - E'_-$ is diamond-free.  Suppose for contradiction that $G'- E'_-$ contains a diamond on $\{x,u,v,y\}$ with cross edge $u v$.  Then $xy\in E(G) \setminus E(G')$, and $\{x,u,v,y\}$ is a clique in $G - E_-$, also in $G$.  The edge $xy$ is not in any type-{\sc i} maximal clique of $G$, and hence $\{x,u,v,y\}$ is a clique of $G$, and part of a type-{\sc ii} maximal clique of $G$.  Hence all edges of this diamond are in $E(G) \setminus E(G')$, a contradiction. 
	
	For the other direction, let $E'_-$ be a minimum solution of $(G', k)$. We claim that $G - E'_-$ is diamond-free.  Suppose for contradiction that $G - E'_-$ contains a diamond on $\{x,u,v,y\}$ with cross edge $u v$.  Then at least one of the five edges of the diamond, i.e., $\{ux, vx, uv, uy, vy\}$, is in $E(G) \setminus E(G')$. Assume without loss of generality that one edge in the triangle on $\{u, v, x\}$ belongs to $E(G) \setminus E(G')$.  Then the maximal clique $K$ of $G$ containing $u,v,x$ is of type {\sc ii}. Now $y$ is adjacent to at least two vertices of $K$ in $G$, hence $y$ must be  in $K$ as well.  But then $xy$ would be in $E(G) \setminus E(G')$, and hence not in $E'_-$; in other words, $x y$ is an edge in $G - E'_-$, a contradiction.
	This concludes the proof.
\end{proof}

After the  application of Rule~\ref{rul:DD-not-part}, all the maximal cliques in the graph are of type {\sc i}.

\begin{redrule}\label{rul:DD-big-t1}
  If there is a vertex $x$ not in any small type-\textsc{i} maximal clique, delete it.
\end{redrule}

\begin{lemma}
  Rule~\ref{rul:DD-big-t1} is safe. 
\end{lemma}
\begin{proof}
	Let $E_-$ be a minimum solution of $(G - x, k)$.  We show that $G^* = G - E_-$ is also diamond-free.  Suppose for contradiction that there are maximal cliques $K_1$ and $K_2$ in $G^*$  such that $|K_1\cap K_2| \ge 2$.  At least one of them contains $x$; assume without loss of generality $x\in K_1$.  Let $K$ be a maximal clique of $G$ containing $K_1$.  By the assumption ($x$ is not in any small maximal clique), $K$ is big.   By Proposition~\ref{pro:untouchable-edge}(i), no edge in the clique $K\setminus \{x\}$ can be deleted by $E_-$. Hence $K = K_1$ and $K_2\not\subseteq K$.  
	By Proposition~\ref{lem:bigT1:deletion}, $K_2$ is small, and $x \notin K_2$. 
	We can find a pair of nonadjacent vertices $x' \in K\setminus \{\{x\} \cup K_2 \}$ and $y \in K_2\setminus K$. If no such pair of vertices exists, then $K_2$ contains all vertices in $K\setminus \{x\}$ and is big, a contradiction. 
	Let $u, v$ be two vertices in $K\cap K_2$.  Then $\{x', u, v, y\}$ induces a diamond in $G - x$ and one of $uy, vy$ has to be deleted by $E_-$, contradicting that $K_2$ is a maximal clique of $G^*$.  This concludes the proof of this lemma.
\end{proof}



It is clear that all the four rules can be applied in polynomial time.  Indeed, a simplified version of procedure \texttt{partition} would suffice.
\begin{lemma}
  Let $(G, k)$ be a yes-instance of the diamond-free edge deletion problem.  If none of  Rules~\ref{rul:DD-1}--\ref{rul:DD-big-t1} is applicable, then $|V(G)| = O(k^3)$.
\end{lemma}
\begin{proof}
	We claim that if neither of Rules~\ref{rul:DD-not-part} and \ref{rul:DD-big-t1} is applicable to a reduced yes-instance $(G,k)$, then $|V(G)| = O(k^3)$. 
	By Lemma~\ref{lem:smallT1-vertex-count}, there are at most $O(k^3)$ vertices in small \tone\ maximal cliques in $G$. After the exhaustive application of Rule~\ref{rul:DD-big-t1}, every vertex is in some small \tone\ maximal clique. Therefore, $G$ contains $O(k^3)$ vertices. 
\end{proof}

\bibliography{diamond-editing}

\end{document}